\documentclass[11pt]{article} \usepackage[article]{optional}

\opt{article}{
    \usepackage{fullpage}
    \usepackage{authblk}
    \usepackage{palatino}
}

\usepackage{cite}
\usepackage{amsmath}
\usepackage{amsfonts}
\usepackage{hyperref}

\hypersetup{
    colorlinks=true,
    linkcolor=blue,
    filecolor=magenta,      
    urlcolor=cyan,
    citecolor=blue,
    pdftitle={Overleaf Example},
    pdfpagemode=FullScreen,
}

\opt{article}{
    \usepackage[capitalise]{cleveref}
}

\usepackage[textsize=scriptsize]{todonotes} 
\usepackage{letltxmacro}
\LetLtxMacro{\todom}{\todo}

\opt{lncs}{
    \setlength{\marginparwidth}{4.3cm} 
}
\opt{article}{
    \setlength{\marginparwidth}{0.9in} 
}

\usepackage{graphicx}
\usepackage{caption}
\usepackage{subcaption}

\usepackage{textcomp} 
\usepackage{listings}
\definecolor{codegreen}{rgb}{0,0.6,0}
\definecolor{codegray}{rgb}{0.5,0.5,0.5}
\definecolor{codepurple}{rgb}{0.58,0,0.82}
\definecolor{backcolour}{rgb}{0.95,0.95,0.92}
\opt{article}{
    \lstset{
        basicstyle=\fontsize{11}{11}\ttfamily
    }
}

\usepackage[normalem]{ulem} 

\newcommand{\R}{\mathbb{R}}

\def\rxn{\mathop{\rightarrow}\limits}  

\def\revrxn{\mathop{\rightleftharpoons}\limits}

\newcommand{\E}[1]{\mathrm{E} \hspace{-0.18em} \left[ #1 \right]}

\renewcommand{\paragraph}[1]{\noindent\emph{#1}}

\usepackage{xspace}
\newcommand{\ppsim}{\texttt{ppsim}\xspace}

\usepackage{amsthm}
\newtheorem{theorem}{Theorem}

\begin{document}

\title{ppsim: A software package for efficiently simulating and visualizing population protocols\thanks{Supported by NSF award 1900931 and CAREER award 1844976.}}

\opt{lncs}{
    \titlerunning{ppsim: A software package for efficiently simulating population protocols}
    
    \authorrunning{D. Doty and E. Severson}
    
    \author{David Doty\inst{1} 
    \and
    Eric Severson\inst{1} 
    }

    \institute{University of California, Davis, CA 95616, USA.
    \email{\{doty,eseverson\}@ucdavis.edu}, 
    }
}

\opt{article}{
    \date{}
    \author[*]{David Doty}
    \author[*]{Eric Severson}
    \affil[*]{University of California, Davis, CA 95616, USA.
        \protect\\
        \texttt{\{doty,eseverson\}@ucdavis.edu}
        \protect\\
        \url{https://web.cs.ucdavis.edu/\~doty/},
        \url{https://eric-severson.netlify.app/}
    }
}

\maketitle

\begin{abstract}
We introduce \ppsim~\cite{ppsim},
a software package for efficiently simulating 
population protocols,
a widely-studied subclass of chemical reaction networks (CRNs) in which all reactions have two reactants and two products.
Each step in the dynamics involves picking a uniform random pair from a population of $n$ molecules to collide and have a (potentially null) reaction.
In a recent breakthrough, Berenbrink, Hammer, Kaaser, Meyer, Penschuck, and Tran~\cite{berenbrink2021simulating}
discovered a population protocol simulation algorithm quadratically faster than the 
na\"{i}ve algorithm,
simulating $\Theta(\sqrt{n})$ reactions in \emph{constant} time 
(independently of $n$, though the time scales with the number of species),
while preserving the \emph{exact} stochastic dynamics.

\ppsim implements this algorithm, with a tightly optimized Cython implementation that can exactly simulate hundreds of billions of reactions in seconds.
It dynamically switches to the CRN Gillespie algorithm for efficiency gains when the number of applicable reactions in a configuration becomes small.
As a Python library, \ppsim also includes many useful tools for data visualization in Jupyter notebooks,
allowing robust visualization of time dynamics such as histogram plots at time snapshots and averaging repeated trials.

Finally, we give a framework that takes any CRN with only bimolecular (2 reactant, 2 product) or unimolecular (1 reactant, 1 product) reactions, with arbitrary rate constants, and compiles it into a continuous-time population protocol. This lets \ppsim exactly sample from the chemical master equation (unlike approximate heuristics such as $\tau$-leaping or LNA), while achieving asymptotic gains in running time.
In linked Jupyter notebooks, we demonstrate the efficacy of the tool on some protocols of interest in molecular programming, including the approximate majority CRN and CRN models of DNA strand displacement reactions.
\end{abstract}

\opt{article}{
}

\section{Introduction}
\label{sec:intro}


A foundational model of chemistry used in natural sciences is that of chemical reaction networks (CRNs)~\cite{Gi77}: 
finite sets of reactions such as 
$A+B \rxn C+D$, representing that molecules $A$ and $B$, upon colliding, can change into $C$ and $D$.
This gives
a continuous time, discrete state, Markov process~\cite{Gi77} modelling discrete counts\footnote{
Another modelling choice are ODEs that describe real-valued concentrations, the ``mean-field'' approximation to the discrete behavior in the large scale limit~\cite{kurtz1972relationship}.
}
of molecules.

Population protocols~\cite{AADFP06},
a widely-studied model of distributed computing with very limited agents,
are a restricted subset of CRNs 
(those with two reactants and two products in each reaction, and unit rate constants) 
that nevertheless capture many of the interesting features of CRNs.
Different terminology is used:
in reaction $A+B \rxn C+D$,
two \emph{agents} (molecules), 
whose \emph{states} (species types) are $A,B$,
have an \emph{interaction} (reaction), changing their states respectively to $C,D$.

\paragraph{Gillespie kinetics for CRNs.}
The standard Gillespie algorithm~\cite{Gi77} simulates 
the Markov process mentioned above.
Given a fixed volume $v \in \R^+$, the \emph{propensity} of a unimolecular reaction 
$r : X \rxn^k \ldots$ is 
$\rho(r) = k \cdot \# X$, where $\# X$ is the count of $X$.
The propensity of a bimolecular reaction 
$r : X + Y \rxn^k \ldots$ is 
$\rho(r) = k \cdot \frac{\# X \cdot \# Y}{v}$ if $X \neq Y$ and 
$k \cdot \frac{\# X \cdot (\# X - 1)}{2v}$ otherwise.
The Gillespie algorithm calculates the sum of the propensities of all reactions: $\rho = \sum_{r} \rho(r)$. 
The time until the next reaction is sampled as
an exponential random variable $T$ with rate $\rho$, and a reaction $r_\text{next}$ is chosen with probability $\rho(r_{\text{next}}) / \rho$ to be applied.


\paragraph{Population protocols.}
The population protocols model comes with simpler dynamics. At each step, a scheduler chooses a random pair of agents (molecules) to interact in a (potentially null) reaction. The discrete time model counts each interaction as $\frac{1}{n}$ units of time, where $n$ is the population size.
A continuous time variant~\cite{fanti2020communication} gives each agent a rate-1 Poisson clock, upon which it interacts with a randomly chosen other agent. 
The expected time until the next interaction is $\frac{1}{n}$, so up to a re-scaling of time, which by straightforward Chernoff bounds is negligible, these two models are equivalent. \ppsim can use either time model.

There is an important efficiency difference between the algorithms: 
the Gillespie algorithm automatically skips null reactions.
For example, a reaction such as $L+L \rxn L+F$, when $\# L = 2$ and $\# F = n-2$, is much more efficient in the Gillespie algorithm,
which simply increments the time until the $L+L \to L+F$ reaction by an exponential random variable in one step.
A na\"{i}ve population protocol simulation iterates through $\Theta(n)$
expected null interactions 
($L+F \to L+F$ and $F+F \to F+F$) until the two $L$'s react.
To better handle cases like this, \ppsim dynamically switches to the Gillespie algorithm when the number of null interactions is sufficiently large;
see documentation~\cite{ppsim} for implementation details.


\paragraph{Other simulation algorithms.}
Variants of the Gillespie algorithm reduce the time to apply a single reaction from $O(|R|)$ to $O(\log|R|)$~\cite{gibson2000efficient} or $O(1)$~\cite{slepoy2008constant}, where $|R|$ is the number of types of reactions. 
However, the time to apply $n$ reactions still scales with $n$.
A common speedup heuristic for simulating $\omega(1)$ reactions in $O(1)$ time is \emph{$\tau$-leaping}~\cite{gillespie2001approximate, rathinam2007reversible, cao2006efficient, gillespie2007stochastic, soloveichik2009robust},
which ``leaps'' ahead by time $\tau$, by assuming
reaction propensities will not change
and updating counts in a single batch step by sampling according these propensities.
Such methods necessarily approximate the kinetics inexactly, though it is possible in some cases to prove bounds on the approximation accuracy~\cite{soloveichik2009robust}.
Linear noise approximation (LNA)~\cite{cardelli2016stochastic} can be used to approximate the discrete kinetics, by adding stochastic noise to an ODE approximation.
A speedup heuristic for population protocol simulation is to sample the number of each interaction that would result from a random matching of size $m$, and update species counts in a single step.
This, too, is an inexact approximation: unlike the true process, it prevents any molecule from participating in more than one of the next $m$ interactions.

The algorithm implemented by \ppsim,
due to 
Berenbrink, Hammer, Kaaser, Meyer, Penschuck, and Tran~\cite{berenbrink2021simulating}, 
builds on this last heuristic.
Conditioned on the event that no molecule is picked twice during the next $m$ interactions, these interacting pairs are a random disjoint matching of the molecules.
Define the random variable $C$ as the number of interactions until the same molecule is picked twice.
Their basic algorithm samples this collision length $C$ according to its exact distribution, 
then updates counts in batch assuming all pairs of interacting molecules are disjoint until this collision,
and finally simulates the interaction involving the collision.
By the Birthday Paradox, $\E{C} \approx \sqrt{n}$ in a population of $n$ molecules,
giving a quadratic factor speedup over the na\"{i}ve algorithm.
The time to update a batch scales quadratically with $q$, the total number of states. The ``multibatch'' variant, used by \ppsim, samples multiple successive collisions to process an even larger batch, and uses $O\bigg(q \sqrt{\frac{\log n}{n}}\bigg)$ time per simulated interaction.

See~\cite{berenbrink2021simulating} for  details.
An advantage of such a fast simulator, specifically for population protocols implementing \emph{algorithms}, 
is that the very large population sizes it can handle (over $10^{12}$) allow one to tell the difference (on a log-scale plot of convergence time)
between a protocol converging in time $O(\log n)$ versus, say, $O(\log^2 n)$.

\section{Usage of the ppsim tool}
\label{sec:usage}

We direct the reader to~\cite{ppsim} for detailed installation, usage instructions, and examples.
Here we highlight basic usage examples for specifying protocols.

There are three ways one can specify a population protocol,
each best suited for different contexts.
The most direct specification of a protocol directly encodes the mapping of input state pairs to output state pairs using a Python \lstinline{dict} (the following is the well-studied \emph{approximate majority} protocol, which has been studied theoretically~\cite{AAE08-2, condon2020approximate} and implemented experimentally with DNA~\cite{chen2013programmable}):

\begin{lstlisting}
a,b,u = 'A','B','U'
approx_majority = {(a,b):(u,u), (a,u):(a,a), (b,u):(b,b)}
\end{lstlisting}

More complex protocols with many possible species are often specified in pseudocode instead of listing all possible reactions.
\ppsim supports this by allowing the \emph{transition function} mapping input states to output states to be computed by a Python function.
The following allows species to be integers and computes an integer average of the two reactants:

\begin{lstlisting}
def discrete_averaging(s: int, r: int):
    return math.floor((s+r)/2), math.ceil((s+r)/2)
\end{lstlisting}

States and transition rules are converted to integer arrays for internal Cython methods, so there is no efficiency loss for the ease of representing protocol rules, since a Python function defining the transition function is not called during the simulation: producible states are enumerated before starting the simulation.

For complicated protocols,
an advantage of \ppsim over standard CRN simulators is the ability to represent species/states as Python objects with different fields (as they are often represented in pseudocode), and to plot counts of agents based on their field values.\footnote{
    Download and run
    \url{https://github.com/UC-Davis-molecular-computing/ppsim/blob/main/examples/majority.ipynb} to visualize 
    such
    large state 
    protocols.
}

Finally, protocols can be specified using CRN-like notation for CRNs with reactions that are bimolecular (2-input, 2-output) or unimolecular (1-input, 1-output), with arbitrary rate constants.
For instance, the CRN
$
\qquad
A + B \revrxn^{0.5}_4 2C, 
\quad\quad 
C \rxn^5 D
\qquad
$
is specified by the code
\begin{lstlisting}
a,b,c,d = species('A B C D')
crn = [(a+b | 2*c).k(0.5).r(4), (c >> d).k(5)]
\end{lstlisting}

This will then get compiled into a continuous time population protocol that samples the same distribution as Gillespie. 
\opt{article}{
    See Section~\ref{sec:full-specification} for details.
}


Any of the three specifications (\lstinline{dict}, Python function, or list of CRN reactions) can be passed to the \lstinline{Simulation} constructor.
The \lstinline{Simulation} can be run to generate a history of sampled configurations.

\begin{lstlisting}
init_config = {a: 51, b: 49}
sim = Simulation(init_config, approx_majority)
sim.run(16, 0.1)   # 160 samples up to time 16
sim.history.plot() # Pandas dataframe with counts
\end{lstlisting}

This would produce the plot shown in Fig~\ref{fig: am1}. When the input is a CRN, \ppsim defaults to continuous time and produces the exact same distributions as the Gillespie algorithm. Fig~\ref{fig: am2} shows a test against the package GillesPy2~\cite{gillespy2} to confirm they sample the same distribution.

\begin{figure}[ht]
    \opt{lncs}{\vspace{-0.5cm}}
    \centering
     \begin{subfigure}[b]{0.46\textwidth}
         \centering
         \includegraphics[width=\textwidth]{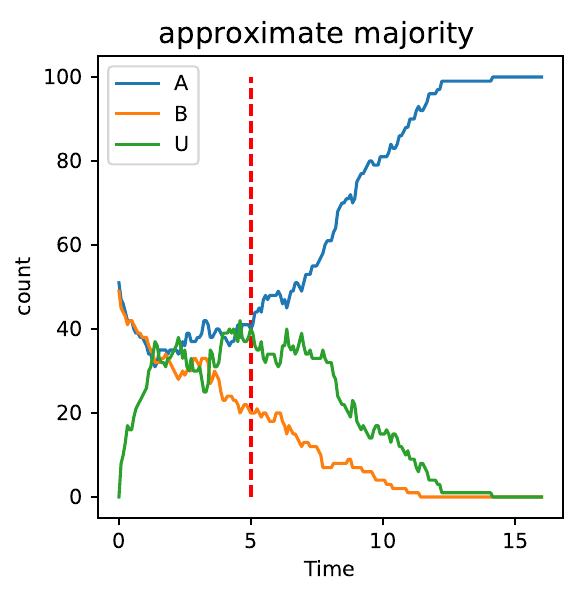}
         \vspace{-0.3cm}
         \caption{\footnotesize Plot of sim.history.}
         \label{fig: am1}
     \end{subfigure}
     \hfill
     \begin{subfigure}[b]{0.52\textwidth}
         \centering
         \includegraphics[width=\textwidth]{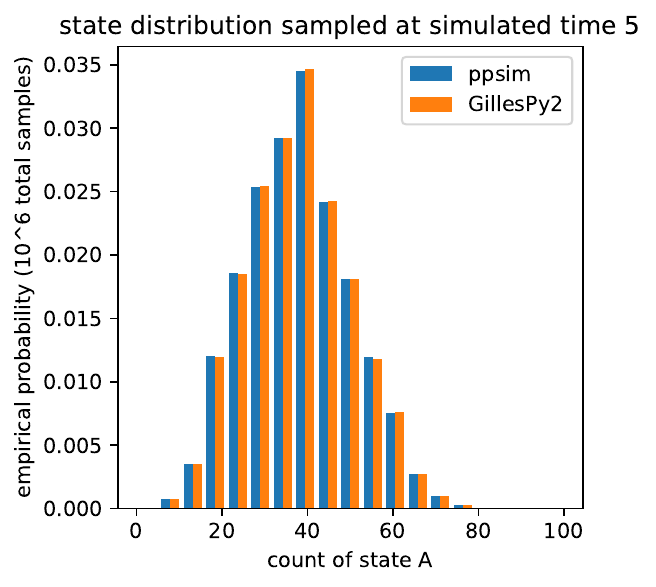}
         \vspace{-0.3cm}
         \caption{\footnotesize Comparison with Gillespie algorithm.}
         \label{fig: am2}
     \end{subfigure}
    \opt{lncs}{\vspace{-0.2cm}}
    \caption{ \opt{lncs}{\footnotesize}
    Time 5 (dotted line in Fig~\ref{fig: am1}) was sampled $10^6$ times with \ppsim and GillesPy2 to verify they both sample the same chemical master equation distribution (Fig~\ref{fig: am2}).
    }
    \label{fig:am-comparison}
    \opt{lncs}{\vspace{-0.2cm}}
\end{figure}

\section{Speed comparison with other CRN simulators}
\label{sec:comparison}

We ran speed comparisons of \ppsim against both GillesPy2~\cite{gillespy2} and StochKit2 \cite{stochkit2},
the latter being the fastest option we found for Gillespie simulation. Fig~\ref{fig:am-runtimes} shows that \ppsim is able to reach significantly larger population sizes.
Other tests shown in an example notebook\footnote{
    \url{https://github.com/UC-Davis-molecular-computing/ppsim/blob/main/examples/crn.ipynb}
    shows further plots and explanations.
} 
show how each package scales with the number of species and reactions.






\begin{figure}[ht]
    \opt{lncs}{\vspace{-0.6cm}}
    \centering
        \opt{lncs}{
            \includegraphics[width=0.6\textwidth]{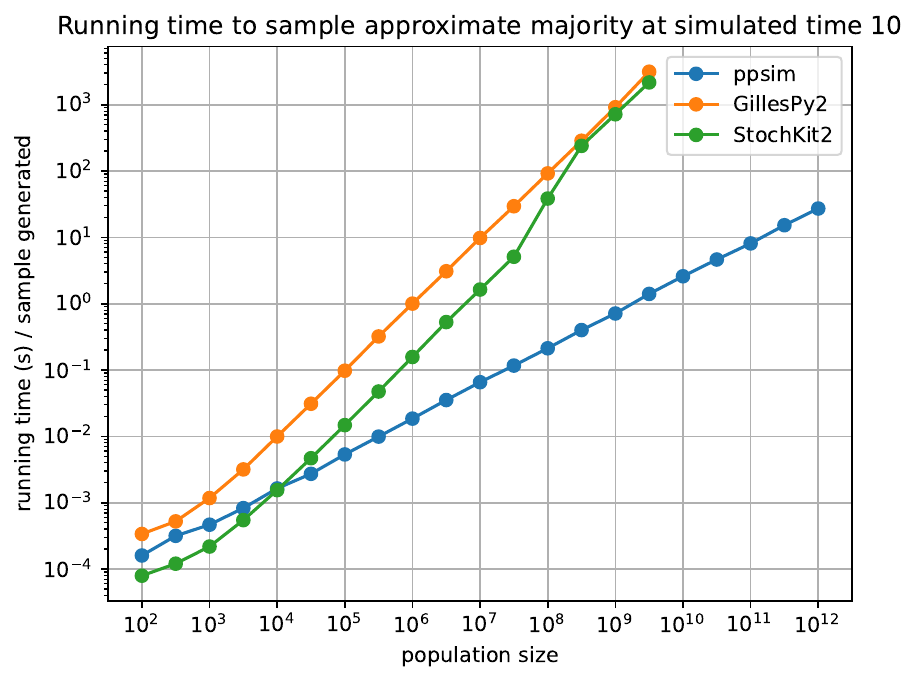}
        }
        \opt{article}{
            \includegraphics[width=0.8\textwidth]{figures/am3.pdf}
        }
    \opt{lncs}{\vspace{-0.4cm}}
    \caption{Comparing runtime with population size n shows $O(n)$ scaling for Gillespie (slope 1 on log-log plot) versus $O(\sqrt{n})$ scaling for \ppsim (slope 1/2).}
    \label{fig:am-runtimes}
    \opt{lncs}{\vspace{-0.9cm}}
\end{figure}

\section{Issues with other speedup methods}
\label{sec:issues}

\begin{figure}[!ht]
    \opt{lncs}{\vspace{-0.2cm}}
    \centering
    \begin{subfigure}[b]{0.24\textwidth}
         \centering
         \includegraphics[width=\textwidth]{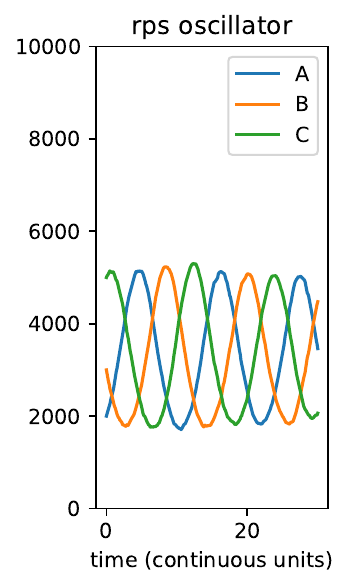}
         \caption{Short timescale oscillations.}
         \label{fig:rps-a}
     \end{subfigure}
    \begin{subfigure}[b]{0.74\textwidth}
         \centering
         \includegraphics[width=\textwidth]{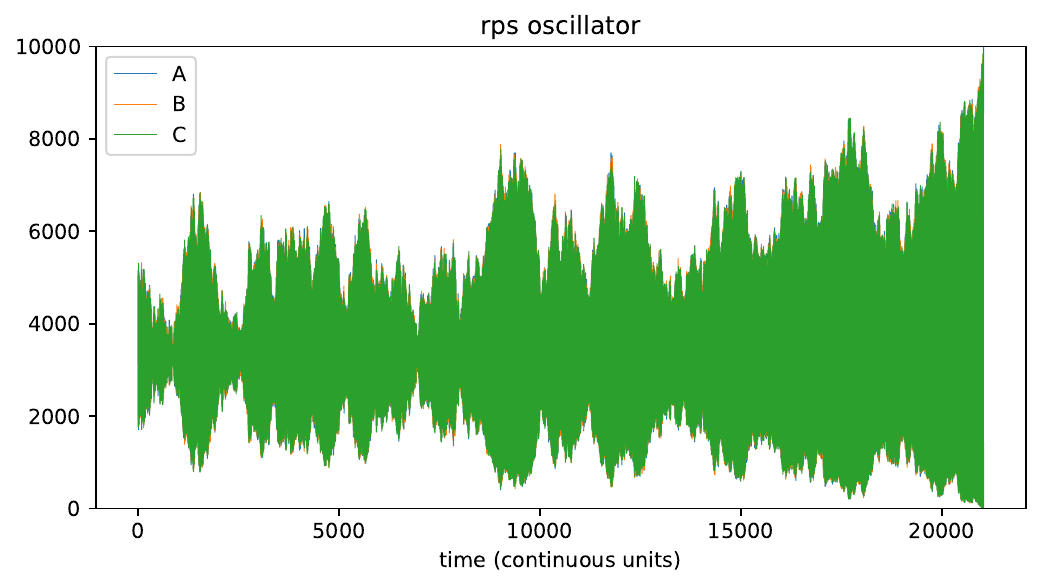}
         \caption{Over a long $\Theta(n)$ timescale, the varying amplitudes will cause two species to go extinct.}
         \label{fig:rps-b}
     \end{subfigure}
     \begin{subfigure}[b]{0.45\textwidth}
         \centering
         \includegraphics[width=\textwidth]{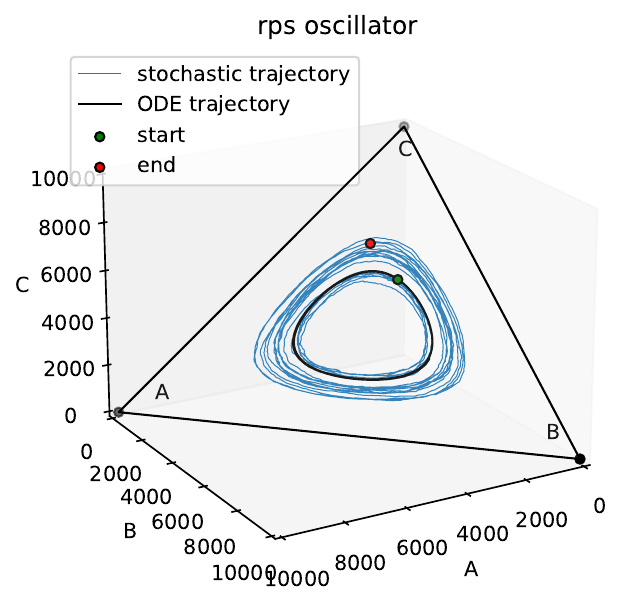}
         \caption{Dynamics from Figs~\ref{fig:rps-a}, \ref{fig:rps-b} in phase space. The ODE solution has a neutrally stable orbit.}
         \label{fig:rps-c}
     \end{subfigure}
     \hfill
     \begin{subfigure}[b]{0.45\textwidth}
         \centering
         \includegraphics[width=\textwidth]{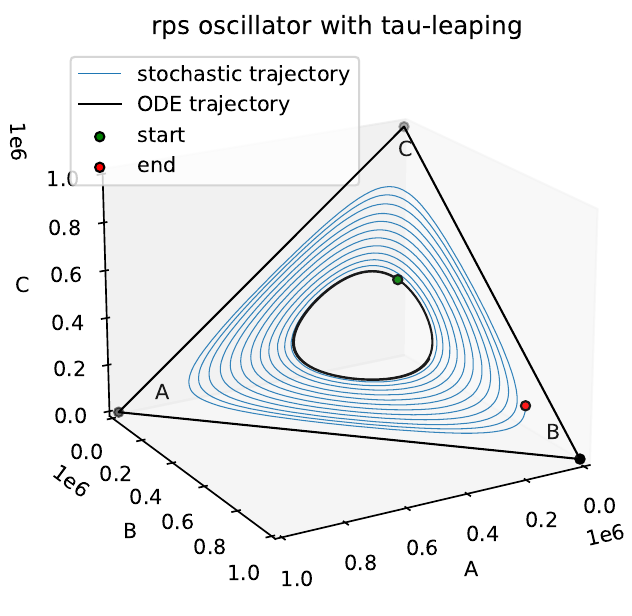}
         \caption{$\tau$-leaping adds a consistent outward drift that will lead to extinction on a much shorter timescale.}
         \label{fig:rps-d}
     \end{subfigure}
    \opt{lncs}{\vspace{-0.2cm}}
    \caption{ \opt{lncs}{\footnotesize}
    The rock-paper-scissors oscillator has qualitative dynamics missed by both ODE simulation (never goes extinct) and $\tau$-leaping (too quickly goes extinct).
    }
    \label{fig:rps}
    \opt{lncs}{\vspace{-0.2cm}}
\end{figure}

It is reasonable to conjecture that exact stochastic simulation of large-count systems is unnecessary, since
Gillespie is fast enough on small-count systems,
and faster ODE approximation is ``reasonably accurate'' for large-count systems.
However, there are example
large count systems with stochastic effects not observed in ODE simulation,
and where $\tau$-leaping introduces systematic inaccuracies that disrupt the fundamental qualitative behavior of the system,
demonstrating the need for exact stochastic simulation.
A simple such example is the 3-state rock-paper-scissors oscillator: 
\begin{lstlisting}
rps = [b+a >> 2*b, c+b >> 2*c, a+c >> 2*a]
\end{lstlisting}
Figure~\ref{fig:rps} compares exact simulation of this CRN to $\tau$-leaping and ODEs.

The population protocol literature furnishes more examples, with problems such as leader election~\cite{alistarh2015polylogarithmic, LeaderElectionDIST, berenbrink2018simple, bilke2017brief, elsasser2018recent, gasieniec2019almost, GS18, berenbrink2020optimal, sudo2020logarithmic, sudo2020leader, burman2021self} and single-molecule detection~\cite{alistarh2017robust, dudek2018universal},\footnote{
    Download and run \url{https://github.com/UC-Davis-molecular-computing/ppsim/blob/main/examples/rps_oscillator.ipynb} to see visualizations of the generalized 7-state rps oscillator used for single-molecule detection in \cite{dudek2018universal}.
}
that crucially use small counts in a very large population, a regime not modelled correctly by ODEs.
See also~\cite{lathrop2020population} for examples of CRNs with qualitative stochastic behavior not captured by ODEs,
yet that behavior appears only in population sizes too large to simulate with Gillespie.
\section{Conclusion}
\label{sec:conclusion}

Unfortunately, the algorithm of Berenbrink et al.~\cite{berenbrink2021simulating} implemented by \ppsim seems inherently suited to population protocols,
not more general CRNs.
For instance, reversible dimerization reactions $A + B \revrxn C$ 
(used, for example, in~\cite{srinivas2017enzyme} to model toehold occlusion reactions in DNA systems)
seem beyond the reach of the batching technique of~\cite{berenbrink2021simulating}.
Although such reactions can be \emph{approximated} by $A+B \revrxn C+F$ for some anonymous ``fuel'' species $F$, the count of $F$ influences the rate of the reverse reaction $F+C \to A+B$, with a different rate than $C \to A+B$.

Another area for improvement 
is the handling of null reactions.
There could be a way to more deeply intertwine the logic of the Gillespie and batching algorithms,
to gain the simultaneous benefits of each, skipping the null reactions
while simulating many non-null reactions in batch.

\bibliographystyle{splncs04}
\bibliography{refs}

\opt{article}{
    \clearpage
    \appendix
\section{Full specification of compilation of CRN to population protocol}
\label{sec:full-specification}

It is possible to specify 1-reactant/1-product reactions such as $A \to B$,
which are compiled into 2-reactant/2-product reactions $A+C \to B+C$ for every species $C$,
with reaction rates adjusted appropriately.
The full transformation is described in the proof of~\cref{thm:sample}.
Here, we give an example of the transformation on the CRN

\begin{align*}
    2A &\revrxn^3_2 B+C
    \\
    C &\rxn^1 D
\end{align*}

First, each reversible reaction is turned into two irreversible reactions:

\begin{align*}
    2A &\rxn^3 B+C
    \\
    B+C &\rxn^2 2A
    \\
    C &\rxn^1 D
\end{align*}

For each non-symmetric bimolecular reaction (with two unequal reactants), add its ``swapped'' reaction reversing the order of reactants and the order of products.
From now on we write reactions using ordered pair notation (e.g., $(A,A) \rxn (B,C)$ instead of $2A \rxn B+C$).

\begin{align*}
    (A,A) &\rxn^3 (B,C)
    \\
    (B,C) &\rxn^2 (A,A)
    \\
    (C,B) &\rxn^2 (A,A)
    \\
    C &\rxn^1 D
\end{align*}

Each (originally) bimolecular reaction (not the result of converting a unimolecular reaction below) has its rates multiplied by the corrective factor $(n-1) / (2 \cdot v)$,
where $n$ is the population size and $v$ is the volume.
We choose $n=v=10$ for this example,
so $(n-1) / (2 \cdot v) = 0.45$.
(See proof of~\cref{thm:sample} for explanation of correction factor.)

\begin{align*}
    (A,A) &\rxn^{1.35} (B,C)
    \\
    (B,C) &\rxn^{0.9} (A,A)
    \\
    (C,B) &\rxn^{0.9} (A,A)
    \\
    C &\rxn^1 D
\end{align*}

Each unimolecular reaction is converted to several bimolecular reactions with all other species in the CRN.

\begin{align*}
    (A,A) &\rxn^{1.35} (B,C)
    \\
    (B,C) &\rxn^{0.9} (A,A)
    \\
    (C,B) &\rxn^{0.9} (A,A)
    \\
    (C,A) &\rxn^1 (D,A)
    \\
    (C,B) &\rxn^1 (D,B)
    \\
    (C,C) &\rxn^1 (D,C)
    \\
    (C,D) &\rxn^1 (D,D)
\end{align*}

Finally, for each ordered pair of input states $(x,y)$,
sum the rates of all reactions that have ordered reactants $(x,y)$, and we let $m$ be the maximum value of this sum over all ordered pairs of reactants.
In this example, the pair $(C,B)$ has rates $0.9$ and $1$ for its two reactions, 
whose sum achieves the maximum $m = 1.9$.
Divide rates by $m$ to convert them to probabilities.

This gives us the final randomized transitions of the population protocol.
Below, whenever the probabilities for a given input state pair $(x,y)$ sum to a value $p < 1$, 
implicitly the transition on input $(x,y)$ is null 
(i.e., outputs $(x,y)$)
with probability $1-p$.

\begin{align*}
    (A,A)&: (B,C)
    \text{ with probability }
    1.35 / 1.9
    \\
    (B,C)&: (A,A)
    \text{ with probability }
    0.9 / 1.9
    \\
    (C,B)&: \{(A,A)
    \text{ with probability }
    0.9 / 1.9
    ,\quad
    (D,B)
    \text{ with probability }
    1 / 1.9
    \}
    \\
    (C,A)&: (D,A)
    \text{ with probability }
    1 / 1.9
    \\
    (C,C)&: (D,C)
    \text{ with probability }
    1 / 1.9
    \\
    (C,D)&: (D,D)
    \text{ with probability }
    1 / 1.9
\end{align*}

Time is now scaled by $m = 1.9$. Thus in one unit of time, there should be an expected $1.9 \cdot n$ interactions. In order to simulate $t$ units of time, we choose a Poisson random variable with mean $1.9 \cdot n \cdot t$ to get the number of interactions to simulate.

    
    
    
    

The following theorem shows that the above transformation results in a population protocol whose continuous time dynamics exactly sample from the same distribution as the Gillespie stochastic model applied to the original CRN.

\begin{theorem}
\label{thm:sample}
Let $\mathcal{C}$ be a CRN consisting of only unimolecular reactions $r_i: X_{i_1} \rxn^{k_i} X_{i_2}$ and bimolecular reactions $r_j: X_{j_1} + X_{j_2} \rxn^{k_j} X_{j_3} + X_{j_4}$.

Then for any initial configuration $\mathcal{I} = \{a_1 X_1, \ldots, a_s X_s\}$ and fixed volume $v \in \R^+$, there exists an equivalent continuous time population protocol $\mathcal{P}$ with time scaling constant $m$. For any time $t\in \R^+$, the distribution over all possible configurations sampled by the Gillespie algorithm at time $t$ is the same distribution as configurations of $\mathcal{P}$ at time $m\cdot t$.
\end{theorem}

\begin{proof}
The continuous time population protocol $\mathcal{P}$ will use the same state set $\{X_1, \ldots, X_s\}$ and initial configuration $\mathcal{I}$. In the population protocol dynamics, each agent has a rate 1 Poisson process for the event where they interact with a randomly chosen agent. Thus for each ordered pair of agents, that pair meets in that order as a rate $\frac{1}{n-1}$ Poisson process.

After converting all reactions, we will be left with a set of 
ordered transitions with rates, of the form $(a,b)\rxn^k (c,d)$. 
An ordinary population protocol transition should correspond to rate $k=1$, so
for each ordered pair of agents $(v_1, v_2)$ in states $(a,b)$, this transition should happen as a Poisson process with rate $\frac{k}{n-1}$.
We must handle the fact that these rates could exceed 1, and also there could be multiple ordered transitions starting from the same pair $(a,b)$.
Define $m$ to be the maximum over all ordered pairs $(a,b)$ of the sum of the rates of any ordered transitions with pair $(a,b)$ on the left. The population protocol transition rule for a pair $(a,b)$ is then a randomized rule, where each ordered reaction $(a,b)\rxn^k (c,d)$ happens with probability $\frac{k}{m}$ (and otherwise the transition is null). Because we are also scaling time by this factor $m$, it follows that the rate of this ordered transition between a single pair of agents $(v_1, v_2)$ in states $(a,b)$ will be $\frac{m}{n-1} \cdot \frac{k}{m} = \frac{k}{n-1}$, as desired.

Next we show how each unimolecular
 reaction
$r_i: X_{i_1} \rxn^{k_i} X_{i_2}$
is converted to a set of ordered transitions with rates. For each $j=1,\ldots,s$, we add the ordered transition $(X_{i_1}, X_j)\rxn^{k_i}(X_{i_2},X_j)$. In other words, the first agent in the pair $v_1$, independent of the state of the other agent, will change state from $X_{i_1}$ to $X_{i_2}$. This agent $v_1$ gets chosen as the first agent in the pair as a Poisson process with rate $m$, and this unimolecular transition will happen (independent of the state of the other agent) with probability $\frac{k_i}{m}$. Thus each agent in state $X_{i_1}$ changes to state $X_{i_2}$ as a rate $k_i$ Poisson process, which is exactly the model simulated by the Gillespie algorithm.

Finally, we show how each bimolecular reaction 
$r_j: X_{j_1} + X_{j_2} \rxn^{k_j} X_{j_3} + X_{j_4}$
is converted to a set of ordered transitions with rates. In the Gillespie model, for each unordered pair $\{v_1, v_2\}$ of agents in states $X_{j_1}$ and $X_{j_2}$, 
the time until this reaction happens is an exponential random variable
with rate $\frac{k_j}{v}$, where $v \in \R^+$ is the volume. 
In our protocol $\mathcal{P}$,
the time when this unordered pair of agents will interact is 
an exponential random variable
with rate $\frac{2m}{n-1}$. Thus, we multiply each rate by the conversion factor $\frac{n-1}{2v}$ to get $k_j' = k_j \cdot \frac{n-1}{2v}$. Then we add the ordered transition $(X_{j_1},X_{j_2}) \rxn^{k_j'} (X_{j_3},X_{j_4})$, and if $X_{j_1} \neq X_{j_2}$, also the reverse ordered transition $(X_{j_2},X_{j_1}) \rxn^{k_j'} (X_{j_4},X_{j_3})$.
As a result, each unordered pair $\{v_1, v_2\}$ will interact with rate $\frac{2m}{n-1}$, then do this transition with probability $\frac{k_j'}{m}$. This gives the reaction a total rate of $k_j' \cdot \frac{2}{n-1} = \frac{k_j}{v}$, as desired.

\end{proof}
}

\end{document}